\documentclass{scrartcl}

\usepackage{amsmath}
\usepackage{amssymb}
\usepackage{amsthm}
\usepackage{xcolor}
\usepackage{xspace}
\usepackage{prettyref}
\usepackage{makeidx}
\usepackage{nomencl}
\usepackage{stmaryrd}
\usepackage{multicol}

\usepackage{tikz}
\usetikzlibrary{trees}
\usetikzlibrary{arrows}

\usepackage{caption}

\usepackage{enumerate}

\usepackage[citecolor=blue,final]{hyperref}

\newtheorem{theorem}{Theorem}[section]
\newtheorem{proposition}[theorem]{Proposition}
\newtheorem{lemma}[theorem]{Lemma}
\newtheorem{corollary}[theorem]{Corollary}
\theoremstyle{definition}

\newtheorem{expl}[theorem]{Example}

\newenvironment{example}{\begin{expl}%
		\pushQED{\qed}}%
	{\popQED\end{expl}}
\newtheorem{defi}[theorem]{Definition}

	{\popQED\end{defi}}

\newtheorem{rema}[theorem]{Remark}
	{\popQED\end{rema}}

\newcommand{\prref}[1]{\prettyref{#1}}
\newrefformat{alg}{Algorithm~\ref{#1}}
\newrefformat{thm}{Theorem~\ref{#1}}
\newrefformat{lem}{Lemma~\ref{#1}}
\newrefformat{def}{Definition~\ref{#1}}
\newrefformat{cor}{Corollary~\ref{#1}}
\newrefformat{prop}{Proposition~\ref{#1}}
\newrefformat{clm}{Claim~\ref{#1}}
\newrefformat{sec}{Section~\ref{#1}}
\newrefformat{subsec}{Subsection~\ref{#1}}
\newrefformat{kap}{Chapter~\ref{#1}}
\newrefformat{ex}{Example~\ref{#1}}
\newrefformat{rem}{Remark~\ref{#1}}
\newrefformat{fig}{Figure~\ref{#1}}
\newrefformat{app}{Appendix~\ref{#1}}
\newrefformat{eq}{Equation~\eqref{#1}}
\newrefformat{tab}{Table~\ref{#1}}
\newrefformat{ap}{{\sc Appendix}}

\newcommand{\abs}[1]{\left|\mathinner{#1}\right|}
\newcommand{\Abs}[1]{\left\lVert\mathinner{#1}\right\rVert}

\newcommand{\set}[2]{\left\{#1\mathrel{\left|\vphantom{#1}\vphantom{#2}\right.}#2\right\}}
\newcommand{\oneset}[1]{\left\{\mathinner{#1}\right\}}
\newcommand\os{\oneset}

\newcommand{\ms}{\ensuremath{\mathrm{ms}}}

\newcommand{\svarietyfont}[1]{\ensuremath{\mathbf{#1}}\xspace}

\newcommand{\varietyfont}[1]{\ensuremath{\mathbf{#1}}\xspace}

\newcommand{\varietyG}{\varietyfont{G}}
\newcommand{\varietyH}{\varietyfont{H}}
\newcommand{\varietyHline}{\overline{\varietyfont{H}}}

\newcommand{\Gcom}{\ensuremath{\varietyfont{Ab}}}

\newcommand{\mediumSize}[1]{\fontsize{9pt}{12pt}\selectfont #1\normalsize}
\newcommand{\mediumFont}[1]{\normalfont\mediumSize{#1}}
\newcommand{\malcev}%
{\mathop{\text{\normalsize{\raisebox{0.3mm}{\textcircled{\raisebox{0.1mm}{\mediumFont{m}}}}}}}}
\newcommand{\Factors}{\mathrm{Factors}}

\newcommand{\ord}{\mathop{\mathrm{ord}}}

\newcommand{\IRR}{\mathrm{IRR}}

\newcommand{\medsmash}[1]{#1}

\newcommand\RAS[2]{\medsmash{\overset{#1}{\underset{#2}{\Longrightarrow}}}}

\newcommand\LAS[2]{\medsmash{\overset{#1}{\underset{#2}{\Longleftarrow}}}}
\newcommand\DAS[2]{\medsmash{\overset{#1}{\underset{#2}{\Longleftrightarrow}}}}

\newcommand\RA[1]{\medsmash{\underset{#1}{\Longrightarrow}}}
\newcommand\LA[1]{\medsmash{\underset{#1}{\Longleftarrow}}}

\newcommand{\kursivdef}[2]{\emph{#1}\index{#2}}
\newcommand{\blinddef}[1]{\index{#1}}

\begin{document}
	\title{Parikh-reducing Church-Rosser representations for some classes of regular languages}
	\author{Tobias Walter\footnote{Supported by the German Research Foundation (DFG) under grant DI 435/6-1.}\\FMI, University of Stuttgart}
	
	\maketitle	
	\begin{abstract}
		In this paper the concept of Parikh-reducing Church-Rosser systems is studied. It is shown that for two classes of regular languages there exist such systems which describe the languages using finitely many equivalence classes of the rewriting system. The two classes are: 1.) the class of all regular languages such that the syntactic monoid contains only abelian groups and 2.) the class of all group languages over a two-letter alphabet.
		The construction of the systems yield a monoid representation such that all subgroups are abelian.
		 Additionally, the complexity of those representations is studied.
	\end{abstract}
	\section{Introduction}
	The class of Church-Rosser congruential languages has been introduced by Narendran, McNaughton and Otto in 1988, see \cite{Narendran84phd,McNaughtonNO88}. A language is Church-Rosser congruential if it is a finite union of equivalence classes of a finite length-reducing Church-Rosser rewriting system. 
	It is natural to ask whether every regular language is Church-Rosser congruential. After some initial progress \cite{Niemann00CRCL, NiemannW02, reinhardtT03, DiekertKW12tcs}, this question has been solved affirmatively, see \cite{DiekertKRW15jacm}.
	The main idea of the solution in \cite{DiekertKRW15jacm} is to prove a stronger statement. Instead of proving that for every regular language there exists a length-reducing Church-Rosser system which saturates the language it is proved that for every regular language and every weight function there exists such a weight-reducing Church-Rosser system. In particular, the initial problem is included by choosing length as the weight function.
	This result on regular languages became possible by utilizing the concept of local divisors. In this paper we use the same technique of local divisors to study a stronger property. Instead of requiring weight-reducing systems for a given weight we ask the question whether for every regular language there exists a Church-Rosser system which saturates the language and is weight-reducing for every weight function. We call such a rewriting system a Parikh-reducing Church-Rosser system. 
	Some of the initial progress already satisfied the Parikh-reducing condition, namely the construction for aperiodic languages \cite{DiekertKW12tcs}, for languages of polynomial density \cite{Niemann00CRCL} and for cyclic groups of order two \cite{NiemannW02}. Our result comprises these results. Namely, the following is the main result: 
	for every language such that its syntactic monoid contains only abelian groups there exists a Parikh-reducing Church-Rosser system which saturates the language. 
	Moreover, all groups appearing in the corresponding Church-Rosser representation are abelian.
	Furthermore, we show the existence of such Parikh-reducing systems for all group languages over a two-letter alphabet.
	Having established the existence of Parikh-reducing systems we study the size of the resulting Church-Rosser representations. Naively, analyzing the construction yields a non-primitive function for this size.
	We introduce an alphabet reduction technique which reduce the size of the resulting Church-Rosser representations to a quadruple exponential function. On the other side of the spectrum we prove an exponential lower bound for cyclic groups.
	\section{Preliminaries}
	\paragraph*{Words and Languages}
	An \kursivdef{alphabet}{alphabet} is a non-empty finite set $A$. 
	An element of $a\in A$ is called a \kursivdef{letter}{letter}.
	A (finite) \kursivdef{word}{word} $w = a_1\cdots a_n$ is a finite concatenation of letters $a_1,\ldots, a_n \in A$. 
	The set of finite words with letters in $A$ is denoted by $A^*$. The \kursivdef{empty word}{word!empty} is denoted by $1$. 
	The set of finite words $A^*$ forms a monoid with the concatenation operation, the \kursivdef{free monoid}{free monoid}. 
    Let $\Abs{\cdot} : A \to \mathbb N$ be a function with $\Abs{a} > 0$ for all $a\in A$. The unique homomorphism, which extends $\Abs{\cdot}$, is also denoted by $\Abs{\cdot}$ and called a \kursivdef{weight}{weight}. 
	A special weight is \kursivdef{length}{length} $\abs{\cdot} : A^* \to \mathbb N$ which is induced by $\abs{a} = 1$ for all $a\in A$. 
	For a letter $c\in A$ we also define $\abs{\cdot}_c : A^* \to \mathbb N$ to be the homomorphism which is induced by \[
	\abs{a}_c = \begin{cases}
	1 & \text{ if } a=c\\ 0 & \text{ else.}
	\end{cases}
	\]
	We set $A^{\leq n} = \set{w\in A^*}{\abs{w} \leq n}$ to be the set of words of length at most~$n$.
	
	A \kursivdef{language}{language} $L$ is a subset of $A^*$. 
	Let $\varphi : A^* \to M$ be a homomorphism in a finite monoid $M$.
	A language $L \subseteq A^*$ is \emph{recognized by $\varphi$} if $L = \varphi^{-1}(\varphi(L))$.
	A language $L$ is regular if it can be recognized by some homomorphism in a finite monoid. 
	\paragraph*{Algebra}
	We want to study subclasses of regular languages which are characterized by special classes of monoids.
	A \kursivdef{variety}{variety} $\varietyfont{V}$ is a class of finite monoids which is closed under taking submonoids, homomorphic images and finite direct products.
	In particular, taking the empty direct product, every variety contains the trivial monoid.
	A variety which contains only groups is called a \kursivdef{variety of groups}{variety!groups}.
	We assign every variety $\varietyfont{V}$ a corresponding language class $\svarietyfont{V}(A^*)$ such that $L \in \svarietyfont{V}(A^*)$ if and only if there exists a monoid $M \in \varietyfont{V}$ and a homomorphism $\varphi : A^* \to M$ that recognizes $L$.
	Examples of such varieties include 
		the variety $\varietyG$ of all groups and the variety $\Gcom$ of all abelian groups.

	Let $\varietyH$ be a variety of finite groups. We define \[
	\varietyHline = \set{M}{\text{every group in } M \text{ is in } \varietyH}
	\]
	to be the maximal class of monoids whose subsemigroups, which are groups, are in $\varietyH$. It turns out that $\varietyHline$ is the maximal variety such that $\overline{\varietyH} \cap \varietyG =
	\varietyH$, see \cite[Proposition V.10.4]{eil76}. Our main result is concerned with the language class $\overline \Gcom (A^*)$.
	An important concept used in this paper are local divisors. 
	Let $M$ be a monoid and $c\in M$. We set $M_c = cM \cap Mc$ and introduce a multiplication $\circ$ on $M_c$ given by
	$uc \circ cv = ucv$.
	Since $uc \in cM$ and $cv \in Mc$, the result of $uc\circ cv$ is in $M_c$. The structure $(M_c,\circ,c)$ forms a monoid, the \kursivdef{local divisor}{local divisor} of $M$ at $c$.
    Indeed, $M_c$ is a divisor of $M$, that is, a homomorphic image of a submonoid of $M$, see \cite{DiekertK2015tcs}. 
	If $c \in M$  is not a unit, then $\abs{M_c} < \abs{M}$ since $1 \not\in cM \cap Mc$.
	
	\paragraph*{Combinatorics on Words}
Let $x = uvw \in A^*$ be a word. 
Then we call $u$ a \kursivdef{prefix}{prefix}, $v$ a \kursivdef{factor}{factor} and $w$ a \kursivdef{suffix}{suffix} of $x$. 
The factor $v$ is \kursivdef{proper}{factor!proper} if $u$ and $w$ are not empty.
The set of factors is given by $\Factors(w) = \set{u}{u \text{ is a factor of } w}$. 
The word $a_1\cdots a_n$, with $a_i \in A$, is a \kursivdef{subword}{subword} of a word~$u$ if $u \in A^* a_1 A^* \cdots A^* a_n A^*$. 
The word $u$ is a power of the word $v$ if $u=v^i$ for some $i\in \mathbb N$. 	
Let $w = a_1\cdots a_n \in A^*$ be a word with $a_i\in A$ letters. We say that $p\in \mathbb N$ is a \kursivdef{period}{period} of $w$ if $a_i = a_{i+p}$ for all $1\leq i \leq n-p$. 
The theorem of Fine and Wilf\blinddef{theorem!Fine and Wilf} describes an important property of periods. 
\begin{theorem}[Fine and Wilf, \cite{FineWilf65}]\label{thm:finewilf}
	Let $p,q$ be periods of some word $w$. If $\abs{w} \geq p+q-\gcd(p,q)$, then $\gcd(p,q)$ is a period of $w$.
\end{theorem}

A word $u$ is called \kursivdef{primitive}{primitive}\blinddef{word!primitive} if it is only a power of itself, that is, if $u=v^i$ with $i\geq 1$ implies $i=1$. The following well-known characterization of primitive words will be useful.
\begin{lemma}\label{lem:charprimitive}
	A word $u\in A^*$ is primitive if and only if $u$ is not a proper factor of $u^2$.
\end{lemma}
\paragraph*{Rewriting systems}
A \kursivdef{semi-Thue system}{semi-Thue system} $S$ over the alphabet $A$ is a finite subset of $A^*\times A^*$. 
An element $(\ell,r) \in S$ is called a \kursivdef{rule}{rule}, where $\ell$ is the left side and $r$ is the right side of the rule. 
The idea of a semi-Thue system is, that left sides of rules can be replaced by right sides of the rule. 
Thus, one often also calls a semi-Thue system a \kursivdef{rewriting system}{rewriting system}. 
For a semi-Thue system $S$ we define the rewriting relation $\RA{S}$ given by
$u_1\ell u_2 \RA{S} u_1ru_2 \text{ for } u_1,u_2\in A^* \text{ and } (\ell,r)\in S$,
that is, $u \RA{S} v$ if $v$ results from $u$ by replacing the left side of a rule with the right side. 
The reflexive transitive closure of $\RA{S}$ is denoted by $\RAS{*}{S}$ and the symmetric closure of $\RAS{*}{S}$ is denoted by $\DAS{*}{S}$. 
We write $v\LA{S} u$ for $u\RA{S} v$.
A semi-Thue system $S$ is
\kursivdef{confluent}{semi-Thue system!confluent} or \emph{Church-Rosser}, if $u \RAS{*}{S} v_1$ and $u \RAS{*}{S} v_2$ imply that there exists a word $w\in A^*$ such that $v_1 \RAS{*}{S} w$ and $v_2 \RAS{*}{S} w$.
	It is \kursivdef{locally confluent}{semi-Thue system!locally confluent}, if $u \RA{S} v_1$ and $u \RA{S} v_2$ imply that there exists a word $w\in A^*$ such that $v_1 \RAS{*}{S} w$ and $v_2 \RAS{*}{S} w$.
It is \kursivdef{weight-reducing}{semi-Thue system!weight-reducing} for a weighted alphabet $(A,\Abs{\cdot})$, if $\Abs{\ell} > \Abs{r}$ for all rules $(\ell, r) \in S$ and it is	
	\kursivdef{Parikh-reducing}{semi-Thue system!Parikh-reducing}, if for all $a\in A$ and all rules $(\ell, r) \in S$ it holds $\abs{\ell}_a \geq \abs{r}_a$ and for all rules $(\ell, r) \in S$ there exists a letter $a\in A$ such that $\abs{\ell}_a > \abs{r}_a$.
	Furthermore, $S$ is \kursivdef{subword-reducing}{semi-Thue system!subword-reducing}, if $r \neq \ell$ and $r$ is a subword of $\ell$ for each rule $(\ell,r)\in S$. 
	
	The notion Parikh-reducing comes from the connection to \kursivdef{Parikh images}{Parikh image}. 
	A Parikh image of a word $w\in A^*$ is the vector $(\abs{w}_a)_{a\in A}$. 
	A semi-Thue system $S$ is Parikh-reducing if and only if 
	the Parikh image $(\abs{r}_a)_{a\in A}$ is smaller than $(\abs{\ell}_{a})_{a\in A}$ for every rule $(\ell,r)\in S$. 
	By definition every subword-reducing system is Parikh-reducing. 
	Further, it is rather easy to see that a semi-Thue system $S\subseteq A^*\times A^*$ is Parikh-reducing if and only if it is weight-reducing for every weight $\Abs{\cdot} : A^* \to \mathbb N$.

	A classical lemma states that $S$ is confluent if it is Parikh-reducing and locally confluent, see \cite{bo93springer}.	
In the following we study different cases which may occur when checking for local confluence. 
				Let $(\ell,r), (\ell',r') \in S$ be two rules and consider the word $u \ell v \ell' w$. Then
				\begin{center}
					\begin{tikzpicture}[implies/.style={double,double equal sign distance,-implies}]	
					\draw (-2.5,0) node (A) {$u \ell v \ell' w$};
					\draw (-2.5,-1.5) node (B) {$u r v \ell' w$};
					\draw (0,0) node (C) {$u \ell v r' w$};
					\draw (0,-1.5) node (D) {$u r v r' w$};
					\draw (A) edge[implies] node[right] {$S$} (B);
					\draw (A) edge[implies] node[below] {$S$}  (C);
					\draw (C) edge[implies] node[right] {$S$} (D);
					\draw (B) edge[implies] node[below] {$S$} (D);
					\end{tikzpicture}
					\end{center}
					Thus, checking for local confluence in this case is trivial. The only non-trivial cases appear when two rules overlap. There are two different kinds of overlaps:
					\begin{enumerate}
						\item $w = x\ell = \ell'y$,
						\item $w = \ell = x\ell'y$
						\end{enumerate}
						for rules $(\ell,r),(\ell',r')\in S$.
						The resulting pairs $(xr,r'y)$ and $(r,xr'y)$ are called \kursivdef{critical pairs}{critical pair}. 
						The first kind is called \kursivdef{overlap critical}{critical pair!overlap critical} and the second kind is called \kursivdef{factor critical}{critical pair!factor critical}, see also \prref{fig:cp}.
						\begin{figure}
							\begin{minipage}[h!]{6.5cm}
								\centering
								\begin{tikzpicture}[implies/.style={double,double equal sign distance,-implies}]
								\def\x{3.5}
								\node [rectangle, thin, draw=black, inner sep = 0pt,minimum height=0.42cm,minimum width = 1.5cm] at (\x-0.25,1.42) {$x\vphantom{\delta^t}$};
								\node [rectangle, thin, draw=black, inner sep = 0pt,minimum height=0.42cm,minimum width = 2cm] at (\x+1.5,1.42) {$\ell\vphantom{\delta^t}$};
								
								\node [rectangle, thin, draw=black, minimum width = 2cm, inner sep = 0pt,minimum height=0.42cm] at (\x,1) 
								{$\ell'\vphantom{\delta^t}$};  
								\node [rectangle, thin, draw=black, minimum width = 1.5cm, inner sep = 0pt,minimum height=0.42cm] at (\x+1.75,1) 
								{$y \vphantom{\delta^t}$};
								\end{tikzpicture}
								
								\caption*{overlap critically}
								\end{minipage}
								\quad\quad
								\begin{minipage}[h!]{5cm}
									\centering
									\begin{tikzpicture}[implies/.style={double,double equal sign distance,-implies}]
									\def\x{3.5}
									\def\y{9}
									\node [rectangle, thin, draw=black, inner sep = 0pt,minimum height=0.42cm,minimum width = 1.5cm] at (\x-0.25,1) {$x\vphantom{\delta^t}$};
									\node [rectangle, thin, draw=black, inner sep = 0pt,minimum height=0.42cm,minimum width = 1.5cm] at (\x+1.25,1) {$\ell'\vphantom{\delta^t}$};
									\node [rectangle, thin, draw=black, inner sep = 0pt,minimum height=0.42cm,minimum width = 1.5cm] at (\x+2.75,1) {$y\vphantom{\delta^t}$};
									
									\node [rectangle, thin, draw=black, minimum width = 4.5cm,inner sep = 0pt,minimum height=0.42cm] at (\x+1.25,1.42) 
									{$\ell\vphantom{\delta^t}$};  
									
									\end{tikzpicture}
									\caption*{factor critically}
									\end{minipage}
									\caption{Sources of critical pairs \cite{DiekertKRW15jacm}}\label{fig:cp}
									\end{figure}
									We say that a critical pair $(u,v)$ resolves if there exists a word $w\in A^*$ such that $u \RAS{*}{S} w \LAS{*}{S} v$ holds. 
									Summarized, we obtain the following:
	\begin{lemma}[\cite{KnuthBendix70}]\label{lem:knuthbendix}
		A semi-Thue system is locally confluent if and only if all its critical pairs resolve.
		\end{lemma}
		\prref{lem:knuthbendix} will be used without explicitly referring to it.

			A word $w$ is \kursivdef{irreducible}{irreducible} in $S$ if no left-side of a rule in $S$ appears in $w$. 
			We denote the set of irreducible elements of $S$ by $\IRR_S(A^*)$. 
			The relation $\DAS{*}{S}$ is a congruence on~$A^*$. Thus, one can consider the monoid $A^*\! /S = A^*\! /\!\!\DAS{*}{S}$. 
			The elements of $A^*\! /S$ are equivalence classes $[u]_S = \set{v\in A^*}{u \DAS{*}{S}v}$ of the congruence $\DAS{*}{S}$. 
			The number of elements in $A^*\! /S$ is called \emph{index} of $S$. 
			If $S$ is Parikh-reducing and (locally) confluent, then 
			there is a bijection between $A^*\! /S$ and $\IRR_S(A^*)$. 
			In this case, we denote elements of the monoid $A^*\! /S$ with the corresponding irreducible words. 
			In fact, we call a locally confluent Parikh-reducing system a \emph{Parikh-reducing Church-Rosser system}.
	Let $\varphi : A^* \to M$ be a homomorphism and $S\subseteq A^*\times A^*$ be a semi-Thue system. 
	We say that $\varphi$ \kursivdef{factorizes through}{factorize through} $S$ if for all $u\RA{S} v$ it holds $\varphi(u) = \varphi(v)$, that is, 
	equivalence classes of $S$ map to the same element in $M$. 
	We also say that $S$ is \emph{$\varphi$-invariant} if $\varphi$ factorizes through $S$. 
	This notion is algebraically motivated. 
	Let $S$ be a semi-Thue system such that $\varphi$ factorizes through $S$, 
	then $\psi : A^*\! /S \to \varphi(A^*)$ given by $\psi([u]_S) = \varphi(u)$ is a well-defined homomorphism. 
	Let $\pi_S : A^* \to A^*\! /S$ be the natural projection and $L$ be some language which is recognized by $\varphi$ and $\pi_L$ be the syntactic homomorphism of $L$. 
	Then we obtain the situation in \prref{fig:facthr}. In particular, $\pi_S$ recognizes $L$.
	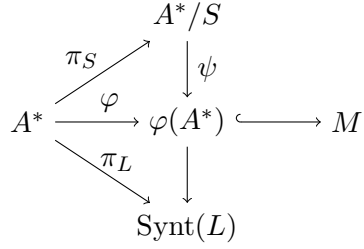
\begin{figure}
		\begin{center}
			\begin{tikzpicture}[scale=0.7]
			\draw (0,0) node (A) {$A^*$};
			\draw (3,0) node (M) {$\varphi(A^*)$};
			\draw (6,0) node (MM) {$M$};
			\draw (3,2) node (S) {$A^* / S$};
			\draw (3,-2) node (L) {$\text{Synt}(L)$};
			\draw[->] (A) -- node[pos=0.65,above] {$\varphi$} (M);
			\draw[->] (A) -- node[above left,outer sep=-2pt] {$\pi_S$} (S);
			\draw[->] (S) -- node[right] {$\psi$} (M);
			\draw[right hook->] (M) -- node[left] {} (MM);
			\draw[->] (A) -- node[pos=0.65,above] {$\pi_L$} (L);
			\draw[->] (M) -- (L);
			\end{tikzpicture}
		\end{center}
		\caption{Algebraic situation of $\varphi$ factorizes through $S$ \cite{DiekertKRW15jacm}}\label{fig:facthr}
	\end{figure}
	Since $\varphi$ factorizes through $S$ if and only if $\varphi : A^* \to \varphi(A^*)$ factorizes through $S$, we may assume that $\varphi$ is surjective. 
	If further $S$ is a Church-Rosser system, we call $A^*\! /S$ a \kursivdef{Church-Rosser representation}{Church-Rosser!representation} of $\varphi$ (or $M$). 

\section{Parikh-reducing Church-Rosser systems}
	\subsection{Outline}
	In this subsection we give an outline on the proof strategy which will be used in \prref{thm:parikhcomgroup}. 
	The macro structure of the proof is as follows: Given a homomorphism $\varphi : A^* \to G$, we construct a system $S$ which is $\varphi$-invariant by induction on $A$. 
	The construction is based on the following lemma:
		\begin{lemma}[\cite{DiekertKW12tcs,DiekertKRW15jacm}]\label{lem:bastel}
		Let $A$ be an alphabet of size at least two, $\varphi : A^* \to M$ be a homomorphism and 
		$B = A\setminus\oneset{c}$ for some $c\in A$. 
		Assume that $R\subseteq B^*\times B^*$ is a Parikh-reducing Church-Rosser system of 
		finite index which is $\varphi$-invariant. Let  $K = \IRR_R(B^*)c$ be a new alphabet and 
		$T \subseteq K^*\times K^*$ be a Parikh-reducing Church-Rosser system of finite index such that 
		\[T' := \set{c\ell \to cr}{\ell\to r \in T}\subseteq A^*\times A^*\] 
		is $\varphi$-invariant. 
		Then 
		\begin{enumerate}[a)]
			\item\label{enum:bastel:a} $S = R\cup T'\subseteq A^*\times A^*$ is a $\varphi$-invariant 
			Parikh-reducing Church-Rosser system of finite index.
			\item\label{enum:bastel:b} All groups in $A^*\! /S$ are contained in $B^*\! /R$ or in $K^*\! /T$.
			\item\label{enum:bastel:c} The index of $A^*/S$ is $\abs{B^*\! /R}+\abs{B^*\! /R}^2\abs{K^*\! /T}$.					
				\end{enumerate}
			\end{lemma}
			\begin{proof}
				\ref{enum:bastel:a}) is proved in \cite{DiekertKRW15jacm}. By \cite{DiekertKW12tcs}, $A^*\! /S$ is a so-called Rees extension monoid and the statement of \ref{enum:bastel:b}) follows from general properties of Rees extension monoids, see \cite{AlmeidaK16}. 
				
				It remains to calculate the size of the index of $S$. Every irreducible word in $S$ which contains no $c$ is contained in $B^*\! /R$. Conversely, every element of $B^*\! /R$ is irreducible in the rewriting system given by $S$. Every irreducible word in $S$ which contains at least one $c$ is of the form $ucvw$ for $u,w\in B^*\! /R$ and $v \in K^*\! /T$. By the definition of the rule set $S$ every such word $ucvw$ is also irreducible. This shows that there are exactly $\abs{B^*\! /R}^2\abs{K^*\! /T}$ irreducible words in $S$ which contains at least one $c$.
			\end{proof}
		For a fixed letter $c\in A$ we remove $c$ and obtain the alphabet $B = A \setminus \oneset{c}$. 
		Inductively, one obtains a system $R \subseteq B^*\times B^*$ which factorizes through $\varphi$. Now, consider a new alphabet $K = \IRR_R(B^*)c$. By \prref{lem:bastel}, it remains to construct a system $T \subseteq K^* \times K^*$. The system $T$ contains two kinds of rules: $\Delta$-rules and $\Omega$-rules. The idea of these rules is to deal with different kind of words. The set $T_\Delta$ of $\Delta$-rules deals with long repetitions of short words. Whenever there is no long repetition of short words, this is witnessed by a marker word $\omega$. The set $T_\Omega$ of $\Omega$-rules contains rules of the form $\omega u \omega \to \omega \gamma(u) \omega$ for some normal forms $\gamma(u)$. \prref{lem:tOmegaexist} shows that such rules appear for sufficiently large words and \prref{lem:parikhisconfl} shows the confluence of the constructed system. 
	\subsection{Commutative Groups}\label{subsec:commgrpcr}
	In this section we study Parikh-reducing Church-Rosser systems for abelian groups. Let $\varphi : A^* \to G$ be a homomorphism in an abelian group $G$. 
	We construct a system for $G$ by sorting the letters $a$ and then reducing them modulo their order. 
	Thus, we actually construct a Church-Rosser representation for the group $\prod_{a\in A} \mathbb Z/\ord(\varphi(a))\mathbb Z$.
	The situation obtained in \prref{thm:parikhcomgroup} is shown in the commutative diagram \prref{fig:parikhcomgroup}.
	\begin{figure}
		\begin{center}
			\begin{tikzpicture}[scale=0.65]
			\draw (-3,2) node (A) {$A^*$};
			\draw (3,0) node (M) {$\varphi(A^*)$};
			\draw (6,0) node (MM) {$G$};
			\draw (3,2) node (S) {$\prod_{a\in A} \mathbb Z/\ord(\varphi(a))\mathbb Z$};
			\draw (3,4) node (SS) {$A^*\! / S$};
			\draw[->] (A) -- node[pos=0.65,below] {$\varphi$} (M);
			\draw[->] (A) -- node[above left,outer sep=-2pt] {$\pi_S$} (SS);
			\draw[->] (A) -- (S);
			\draw[->] (S) -- (M);
			\draw[->] (SS) -- (S);
			\draw[right hook->] (M) -- node[left] {} (MM);
			\end{tikzpicture}
		\end{center}
		\caption{Commutative diagram in the situation of \prref{thm:parikhcomgroup}.}\label{fig:parikhcomgroup}
	\end{figure}
	
	\begin{theorem}\label{thm:parikhcomgroup}
		Let $\varphi : A^* \to G$ be a homomorphism to a finite commutative group $G$. Then there exists a Parikh-reducing Church-Rosser system $S$ of finite index which factorizes through $\varphi$. Further, all groups contained in $A^*\! /S$ are isomorphic to some subgroup of $\prod_{a\in A} \mathbb Z/\ord(\varphi(a))\mathbb Z$.
	\end{theorem}
	\begin{proof}	
		Let $n$ be the least common multiple of $\ord(\varphi(a))$ for $a\in A$.
		We do an inductive proof on the number of letters $\abs A$. If $A = \os{c}$, then we may set $S = \os{c^n \to 1}$. 
		This system is Parikh-reducing, locally confluent and it holds $A^*\! /S \simeq \mathbb Z/n\mathbb Z$. Thus, we may assume that $\abs A > 1$. Let $A = \os{a_1, \ldots, a_s, c}$ be the alphabet and $c \in A$ be an arbitrary letter of $A$. We consider the alphabet $B = A \setminus \os{c}$. Inductively, $B$ is smaller than $A$, we get a Parikh-reducing Church-Rosser system $R \subseteq B^* \times B^*$ of finite index which factorizes through $\varphi_{|B^*} : B^* \to G$. 
		The idea is to first reduce the words over~$B^*$ and then work over a new alphabet $K$. Let $K = \IRR_R(B^*)c$ be the new alphabet of irreducible words over $B^*$ appended by the letter $c$ which poses as a separator.
		We will first construct a Parikh-reducing (over $A^*$)  Church-Rosser system $T \subseteq K^* \times K^*$ of finite index. Note that this system $T$ is not Parikh-reducing over $K^*$.
		We will use two different sets of rules. One for long repetitions of short words and one for longer words which are not repetitions of such short words. 
		%
		Let us first define the set of short words as $\Delta = K^{\leq n}\setminus\os{1}$, that is, as the set of nonempty words of length at most $n$.
		Let further be $$T_\Delta = \set{\delta^{t+n} \to \delta^t}{\delta \in \Delta}$$ the system of $\Delta$-rules whereas $t = 3n(s+4)+n$. 
		The choice of the parameter $t$ will be explained later. 
		For now, the fact that $t > 2n$ is sufficient to obtain that 
		$T_\Delta$ is a Parikh-reducing (over $K^*$, and thus also over $A^*$) Church-Rosser system by  \prref{lem:deltacrsystem}. 
			\begin{lemma}[\cite{DiekertKRW15jacm}]\label{lem:deltacrsystem}
				Let $\Delta \subseteq K^{\leq n}$ be a set of nonempty words of length at most $n$ which is closed under nontrivial factors, $t > 2n$ and $n\geq 1$.
				Then \[
				T_\Delta = \set{\delta^{t+n} \to \delta^t}{\delta\in \Delta}
				\] is a subword-reducing Church-Rosser system. In particular, $T_\Delta$ is Parikh-reducing and weight-reducing for every weight.
			\end{lemma}
		%
		Next, we will introduce marker words. They basically mark the absence of a long repetition of words in $\Delta$, i.e., a long enough word in $K^*$ will either contain a marker word or a rule in $T_\Delta$. The next lemma shows that the length of such markers can be bounded by $2n$.
			\begin{lemma}[\cite{DiekertKRW15jacm}]\label{lem:minimalefactorgegenbsp}
				Let $\Delta \subseteq K^{\leq n}$ be a set and let $F = \bigcup_{\delta \in \Delta, i \in \mathbb N} \Factors(\delta^i)$. 
				Then $K^* \setminus F$ is an ideal which is generated by a set $J \subseteq K^{\leq 2n}$ of words of length at most $2n$, that is, $K^* \setminus F = K^* J K^*$.
			\end{lemma}
		Thus, letting $F = \bigcup_{\delta \in \Delta, i \in \mathbb N} \Factors(\delta^i)$, we obtain $K^* \setminus F = K^* J K^*$ for some $J \subseteq K^{\leq 2n}$.
		In order to ensure that we find such a marker which does not start with a $c \in K$, we increase the length of a marker to $3n$. Formally, let $\Omega = K^{3n} \setminus (cK^* \cup F)$ be the set of markers.

		%
		Let $\preceq$ be a total preorder on $\Omega$ with the following properties:
		\begin{itemize}
			\item $\omega, \eta \in\Omega$ with 
			$\omega \in K^*(K\setminus\os{c}) c^i, \eta \in K^*(K\setminus\os{c}) c^j$ and $i>j$ 
			implies $\omega \preceq \eta$.
			\item $\preceq$ is a total order on $\Omega \setminus Kc^{3n-1}$.
			\item $\omega, \eta \in \Omega\cap Kc^{3n-1}$ implies $\omega \preceq \eta$.
		\end{itemize}
		Thus, the larger the block of $c$'s at the suffix of an $\omega$, the smaller it is with respect to~$\preceq$.
		Additionally, all elements in $\Omega$ with a maximal block of $c$'s at the suffix are equivalent with respect to $\preceq$. 
		In particular, $\omega \preceq \eta$ and $\eta \preceq \omega$ implies either $\eta = \omega$ or there exists $b_1,b_2 \in K$ with $\omega = b_1 c^{3n-1}$ and $\eta = b_2c^{3n-1}$. 
		Let $u \in K^* \omega K^*$ for some $\omega \in \Omega$. 
		We say that $\omega$ is a \kursivdef{maximal $\Omega$-factor}{maximal!$\Omega$-factor} of $u$, 
		if $u \in K^* \eta K^*$ with $\eta \in \Omega$ implies $\eta \preceq \omega$.
		We want to show that every long word contains sufficiently large factors which are surrounded by ``locally'' maximal $\Omega$-factors. The first step is to show the existence of $\Omega$-factors.
		\begin{lemma}\label{lem:t0parikhexist}
			There exists a number $t_0$ such that for every word $v \in K^*$ with length at least 
			$t_0$ has a factor $\delta^{t+n}$ for some $\delta \in \Delta$ or a factor $\omega \in \Omega$.
		\end{lemma}
		\begin{proof}
			Let $t_0 = (t+n+3)(n+1)$. 
			If $v \notin \IRR_{T_\Delta}(K^*)$ the statement is true. Thus, we assume that for all $\delta \in \Delta$ there is no factor $\delta^{t+n}$ of $v$. 
			There is a factorization $v = c^\ell v_1 v_2$ such that $v_1 \in F$ is maximal and $v_1$ has no $c$ as a prefix.
			\begin{figure}			 
				\begin{center}\begin{tikzpicture}[implies/.style={double,double equal sign distance,-implies}]
					\def\x{0}
					\def\y{9}		 	
					\node [rectangle, thin, draw=black, minimum width = 0.6cm, inner sep = 1pt] at (\x+1.65,1.53) {$c^\ell\vphantom{c^\ell_1}$};  
					\node [rectangle, thin, draw=black, minimum width = 1.6cm, inner sep = 1pt] at (\x+2.75,1.53) 
					{$v_1\vphantom{c^\ell_1}$};
					\node [rectangle, thin, draw=black, minimum width = 2.5cm, inner sep = 1pt] at (\x+4.8,1.53) 
					{$v_2\vphantom{c^\ell_1}$};
					
					\node [rectangle, thin, draw=black, inner sep = 1pt,minimum width = 0.5cm] at (\x+2.2,1) {$\delta\vphantom{c^\ell_1}$};
					\node [rectangle, thin, draw=black, inner sep = 1pt,minimum width = 0.5cm] at (\x+2.7,1) {$\delta\vphantom{c^\ell_1}$};
					\node [rectangle, thin, draw=black, inner sep = 1pt,minimum width = 0.5cm] at (\x+3.2,1) {$\delta\vphantom{c^\ell_1}$};
					
					\node [rectangle, thin, draw=black, inner sep = 1pt,minimum width = 1cm] at (\x+3.85,0.47) {$u\vphantom{c^\ell_1}$};
					\node [rectangle, thin, draw=black, inner sep = 1pt,minimum width = 1.4cm] at (\x+3.65,-0.05) {$u'\vphantom{c^\ell_1}$};
					\node [rectangle, thin, draw=black, inner sep = 1pt,minimum width = 1.6cm] at (\x+3.75,-0.58) {$u''\vphantom{c^\ell_1}$};
					\end{tikzpicture}
				\end{center}
				\caption{Construction of a factor in $\Omega$ as used in \prref{lem:t0parikhexist}.}
			\end{figure}
			Hence we obtain $\ell < t+n$ and $\abs{v_1} < (t+n)n$ which implies $\abs{v_2} \geq 3n+3 > 3n-1$ by definition of $t_0$. 
			As $v_1 \in F$, there is some $\delta \in \Delta$ which does not have $c$ as prefix and $v$ is a prefix of $\delta^+$. 
			Consider the first factor $u$ of length $2n$ of $v_1v_2$ which is not in $F$. 
			Since $v_1$ is a prefix of $\delta^+$, one must take at most $n-1$ additional letters left from $u$ in order to obtain a factor $u'$ of $v_1v_2$ which is not in $F$, has length at most $3n$ and does not start with a $c$. 
			Filling up $u'$ with letters from the right, we obtain a factor $u''$ of $v_1v_2$ which is not in $F$, has length $3n$ and does not start with a $c$, that is, $u'' \in \Omega$.
		\end{proof}
			\begin{lemma}\label{lem:tOmegaexist}
				There exists a number $t_\Omega$ such that every word $v\in K^*$ of length at least $t_\Omega$ contains 
				either
				\begin{itemize}
					\item  a factor $\delta^{t+n}$ for $\delta\in \Delta$ or 
					\item a factor $\omega u \omega$ with $\omega \in \Omega$, $t < \abs{\omega u \omega} \leq t_\Omega$ and for every $\eta \in \Omega$ 
					with $\omega u \omega\in A^* \eta A^*$ we have $\eta \preceq \omega$.
				\end{itemize}
			\end{lemma}
			\begin{proof}
				Let $\Omega_v = \set{\omega \in \Omega}{v \in A^* \omega A^*}$ be the set of $\Omega$-factors of $v$ and 
				let $t_k$ be defined by the recursion $t_k = 2t_{k-1} + t$. 
				A quick calculation verifies the explicit formula $t_k = 2^k(t_0+t)-t$.
				We prove the following statement by induction on $k$:
				For every word $v$ of length at least $t_k$ which has at least $k$ different $\Omega$-factors, i.e., $k \geq \abs{\Omega_v}$ and which does not contain a factor $\delta^{t+n}$ for $\delta \in \Delta$, there exists a factor $\omega u \omega$ of $v$ such that
				\begin{itemize}
					\item $\omega \in \Omega$,
					\item $t < \abs{\omega u \omega} \leq t_k$ and
					\item $\omega$ is a maximal $\Omega$-factor of $\omega u \omega$.
				\end{itemize}
				The case $k=0$ is trivial since by hypothesis every word $v$ with length at least $t_0$ and $\abs{\Omega_v} = 0$ must contain a factor $\delta^{t+n}$ for $\delta \in \Delta$.
				Consider the case $k>0$. Since we require that the length of the factor $\omega u \omega$ is smaller or equal to $t_k$, 
				we consider the prefix of $v$ of length $t_k$. 
				In particular, 
				we can assume that every proper factor of $v$ has length smaller than $t_k$.
				
				Consider the factorization $v = pfq$ with $f \in (\omega A^* \cap A^* \omega)$ such that $\omega$ is a maximal $\Omega$-factor of $v$ and $f$ is maximal with regard to length. 
				If $\abs f \leq t$, we obtain 
				$$t_k = 2t_{k-1} + t \leq \abs{pfq} = \abs{pq} + \abs f \leq \abs{pq} + t$$ 
				which implies $\abs{pq} \geq 2t_{k-1}$.
				Since $p$ and $q$ contain no factor $\omega$, we can apply induction to either $p$ or $q$.
				If $\abs f > t$, then $f$ has the form $f = \omega u \omega$ for a word $u$ because of $t> 2\max_{\omega \in \Omega}\abs{\omega}$ and $f \in (\omega A^* \cap A^* \omega)$. The factor $f$ has the required properties since $\abs f \leq \abs v \leq t_k$. This concludes the induction.
				We infer the statement of the lemma by setting $t_\Omega = t_{\abs \Omega}$.
			\end{proof}
			
		In particular, \prref{lem:tOmegaexist} shows the existence of a number $t_\Omega$ such that every $v \in \IRR_{T_\Delta}(K^*)$ 
		with $\abs v \geq t_\Omega$ contains a factor $\omega u  \omega'$ 
		with $\omega, \omega'$ being $\Omega$-maximal for this factor and $t< \abs{\omega u \omega'} \leq t_\Omega$. 
		The idea is to reduce $u$ to a normal form $\gamma(u)$. This is the part where commutativity of $G$ is needed.
		Let $a\in A$ be a letter and $\abs{u}_a$ be the number of occurrences of $a$ in $u$. 
		Define $\gamma_a(u) = a^{\abs{u}_a \mod \ord(\varphi(a))} c^{3n}$ and $$\gamma(u) = c^{3n} \gamma_{a_1}(v)\ldots \gamma_{a_s}(v)\gamma_c(v).$$ 
		The mapping $\gamma$ is a normal form in the group $\prod_{a \in A} \mathbb Z/\ord(a)\mathbb Z$, i.e., let $\psi : A^* \to \prod_{a \in A} \mathbb Z/\ord(a)\mathbb Z$ be the homomorphism counting the different letters $a$ modulo $\ord(a)$, then $\psi(u) = \psi(v)$ if and only if $\gamma(u) = \gamma(v)$.
		By choice of $\gamma_a(u)$ we have $\gamma(u) \in K^*$. Since $\abs{\gamma_a(u)} = 3n$ for $a\in B$ and $3n \leq \abs{\gamma_c(u)} < 4n$, we obtain \[t-7n = 3n(s+2) \leq \abs{\gamma(u)} < 3n(s+2)+n = t - 6n.\]
		In particular, $\varphi(u) = \varphi(\gamma(u))$ and $\gamma(u\gamma(u')) = \gamma(uu') = \gamma(u'u) = \gamma(\gamma(u')u)$.
		Additionally, if $u\in K^*$ with $\abs{u} \geq 3n(s+2)+n = t -6n$, then $u \mapsto \gamma(u)$ is Parikh-reducing over $A^*$ since at least the number of $c$ decreases. 
		Note that the inequality $t-n \leq \abs{\omega \gamma(u) \omega'} < t$ is actually the reason for the definition of $t$.
		Let $$T_\Omega = \set{\omega u \omega' \to \omega \gamma(u) \omega'}{t\leq \abs{\omega u \omega'} \leq t_\Omega \text{ and } \omega, \omega' \text{ are }\Omega\text{-maximal for } \omega u \omega'}$$ 
		be the set of $\Omega$-rules. By definition of $\gamma$ the set of $\Omega$-rules is Parikh-reducing over~$A^*$. 
		Note that for a $\Omega$-rule, either $\omega$ and $\omega'$ are minimal elements in $\Omega$ or $\omega = \omega'$.
		By \prref{lem:tOmegaexist} the system $T = T_\Delta \cup T_\Omega$ has only finitely many irreducible elements. 
		It remains to prove that $T$ is Church-Rosser. By \prref{lem:deltacrsystem} the set $T_\Delta$ of $\Delta$-rules is (locally) confluent. 
		Next, we will study properties of $\Omega$-rules which are crucial for showing that $T$ is Church-Rosser.
		First, we show that $T$-rules preserve $\Omega$-maximal elements.
		\begin{lemma}\label{lem:omegamax}
			Let $u \RA{T} v$ and let $\omega$ be a maximal $\Omega$-factor of $u$. 
			Then $\eta \preceq \omega$ for every  $\Omega$-factor $\eta$ of $v$.
		\end{lemma}
		\begin{proof}
			As $T = T_\Delta \cup T_\Omega$ there are two cases for the rule set of $u \RA{T} v$.
			
			In the case that $u \RA{T_\Delta} v$ there must exists a $\delta \in \Delta$ and a factorization $u = u_1 \delta^{t+n} u_2$ such that $v = u_1 \delta^t u_2$. 
			By construction, we have  $t > 3n = \abs{\omega}$. Thus, every element of $\Omega$ is a factor of $u$ if and only if it is also a factor of $v$. 
			Since $\omega$ is $\Omega$-maximal for $u$, it is also $\Omega$-maximal for $v$.
			
			If $u \RA{T_\Omega} v$, there is a factorization
			$u = u_1 \omega_1 \hat u \omega_2 u_2$ such that $v = u_1 \omega_1 \gamma(\hat u) \omega_2 u_2$ and $\omega_1, \omega_2$ are maximal $\Omega$-factors of $\omega_1 \hat u \omega_2$. 
			Since every marker in $\Omega$ has fixed length $3n$, it remains to show that $\omega_1 \gamma(\hat u) \omega_2$ has no $\Omega$-factors larger than $\omega_1$ (and by $\omega_1 \preceq \omega$, also no $\Omega$-factors larger than $\omega$).
			Note that $\gamma(\hat u)$ has $c^{3n}$ as prefix and suffix. 
			Every $\Omega$-factor of $\omega_1 \gamma(\hat u)$ which is not an $\Omega$-factor of $\gamma(\hat u)$ has the form $\zeta c^i$ for some $i\geq 0$ and $\zeta$ is a suffix of $\omega_1$. Since the block of $c$'s at the suffix of $\zeta c^i$ may only increase, we obtain $\zeta c^i \preceq \omega_1$ by definition of $\preceq$. 
			Since every element of $\Omega$ has length $3n$ and does not have $c$ as a prefix, 
			there is no $\Omega$-factor in $\gamma(\hat u) \omega_2$ which is neither in $\gamma(\hat u)$ nor equals $\omega_2$. 
			By construction, every $\Omega$-factor of $\gamma(\hat u)$ is of the form $\gamma_a(\hat u)$ for some $a\in B$. 
			However, $\gamma_a(\hat u)$ is a minimal element of $\Omega$ by construction. 
			In particular, $\eta \preceq \omega_1 \preceq \omega$ for every $\Omega$-factor $\eta$ of $\omega_1 \gamma(\hat u) \omega_2$.
		\end{proof}
		
		Next, as an intermediate step, we show local confluence in the case of a left side $\omega u \omega'$ of a rule in $T_\Omega$. In particular, we show that every word of this form can be reduced to a fixed normal form.
		\begin{lemma}\label{lem:parikhnf}
			Let $\omega u \omega'$ be a word such that $\omega$ and $\omega'$ are maximal $\Omega$-factors of $\omega u \omega'$ and $\abs{\omega u \omega' }\geq t$. Then $\omega u \omega' \RA{T} v$ implies $v \RAS{*}{T} \omega \gamma(u) \omega'$.
		\end{lemma}
		\begin{proof}
			The statement is clear if $v = \omega \gamma(u) \omega'$ which is why we may assume $v \neq \omega \gamma(u) \omega'$. 
			We show the lemma inductively on the length of $\omega u \omega'$. In order to apply the induction step we show that $v = \omega v' \omega'$ and $\abs v \geq t$. The precondition that $\omega$ and $\omega'$ are maximal $\Omega$-factors of $v$ is satisfied by \prref{lem:omegamax}.
			
			In the case of $\omega u \omega' \RA{T_\Omega} v$, some rule $\mu u' \mu' \to \mu \gamma(u') \mu' \in T_\Omega$ was applied. 
			As such rules preserve the prefixes and suffixes of length $3n$, the word $v$ must have the correct form. 
			In the case of $\omega u \omega' \RA{T_\Delta} v$, some rule $\delta^{t+n} \to \delta^t$ was applied. 
			Since $t>6n$ and elements of $\Omega$ all have length $3n$, the $\Omega$-factors $\omega$ and $\omega'$ are preserved by the application of the $\Delta$-rule $\delta^{t+n} \to \delta^t$.
			In both cases we conclude that $v = \omega v' \omega'$ for some word $v'$.
			
			It remains to show, that $\abs v \geq t$. Since $\abs{\delta^t} \geq t$, the case of an application of a rule in $T_\Delta$ is trivial. 
			Let $v$ stem from the application of a rule $\mu u' \mu' \to \mu \gamma(u') \mu' \in T_\Omega$.
			If either $\mu u'$ or $u' \mu'$ is a factor of $u$, we have that either $\mu \gamma(u')$ or $\gamma(u') \mu'$ is a factor of $v'$. 
			Thus, using $\abs{\gamma(u')} > t - 7n$ and $\abs{\omega} = 3n$ for every element $\omega\in \Omega$, we obtain \[\abs{v} = \abs{\omega v' \omega'} \geq \abs{\omega} + \abs{\mu \gamma(u')} + \abs{\omega'} > t + 2n > t.\]
			
			It remains to prove $\abs v \geq t$ for the situation which is depicted below.
			\begin{center}\begin{tikzpicture}[implies/.style={double,double equal sign distance,-implies}]
				\def\x{3.5}
				\def\y{9}
				\node [rectangle, thin, draw=black, minimum width = 1cm, inner sep = 1pt, minimum height=0.4cm] at (\x-0.25,1.48) 
				{$\omega\vphantom{\mu\delta^t}$};  
				\node [rectangle, thin, draw=black, minimum width = 2.5cm, inner sep = 1pt, minimum height=0.4cm] at (\x+1.5,1.48) 
				{$u\vphantom{\mu\delta^t}$};  
				\node [rectangle, thin, draw=black, minimum width = 1cm, inner sep = 1pt, minimum height=0.4cm] at (\x+3.25,1.48) 
				{$\omega'\vphantom{\mu\delta^t}$};

				\node [rectangle, thin, draw=black, inner sep = 1pt,minimum width = 1cm, minimum height=0.4cm] at (\x+0.25,1) {$\mu\vphantom{\mu\delta^t}$};
				\node [rectangle, thin, draw=black, inner sep = 1pt,minimum width = 1.5cm, minimum height=0.4cm] at (\x+1.5,1) {$u'\vphantom{\mu\delta^t}$};
				\node [rectangle, thin, draw=black, inner sep = 1pt,minimum width = 1cm, minimum height=0.4cm] at (\x+2.75,1) {$\mu'\vphantom{\mu\delta^t}$};
				
				\end{tikzpicture}
			\end{center}
			If $\omega \neq \omega'$, then there exists $b_1,b_2 \in K\setminus\oneset{c}$ such that $\omega = b_1 c^{3n-1}$ and $\omega' = b_2c^{3n-1}$. 
			However, as no element of $\Omega$ starts with the letter $c$, we can conclude $\omega = \mu$ and thus by $\mu' \preceq \mu$ we obtain $\omega' = \mu'$ by the same argument. 
			In this case we have $\omega u \omega' = \mu u' \mu'$ and henceforth $v = \omega \gamma(u) \omega'$. 
			The case that $\mu \neq \mu'$ is similar: $\omega'$ has no $c$ as prefix and thus $\mu' = \omega'$. Again, $\omega = \mu$ and $v = \omega \gamma(u) \omega'$ holds.
			Hence, we may assume $\omega = \omega'$ and $\mu = \mu'$. 
			
			Combining both overlaps, we obtain the following picture.
			\begin{center}\begin{tikzpicture}[implies/.style={double,double equal sign distance,-implies}]
				\def\x{3.5}
				\def\y{9}
				\node [rectangle, thin, draw=black, minimum width = 0.75cm, inner sep = 1pt] at (\x-0.375,1.48) 
				{$x\vphantom{\mu\delta^t}$};  
				\node [rectangle, thin, draw=black, minimum width = 2cm, inner sep = 1pt] at (\x+1,1.48) 
				{$\omega\vphantom{\mu\delta^t}$};  
				\node [rectangle, thin, draw=black, minimum width = 0.75cm, inner sep = 1pt] at (\x+2.375,1.48) 
				{$y\vphantom{\mu\delta^t}$};

				\node [rectangle, thin, draw=black, inner sep = 1pt,minimum width = 2cm] at (\x+0.25,1) {$\mu\vphantom{\mu\delta^t}$};
				\node [rectangle, thin, draw=black, inner sep = 1pt,minimum width = 1.5cm] at (\x+2,1) {$y'\vphantom{\mu\delta^t}$};

				\node [rectangle, thin, draw=black, inner sep = 1pt,minimum width = 1.5cm] at (\x,0.52) {$x'\vphantom{\mu\delta^t}$};
				\node [rectangle, thin, draw=black, inner sep = 1pt,minimum width = 2cm] at (\x+1.75,0.52) {$\mu\vphantom{\mu\delta^t}$};
				
				\end{tikzpicture}
			\end{center}
			In the notation of the picture above we have $u = y u' x$. Thus, $v = \omega y \gamma(u') x \omega$ and by $\gamma(u') > t - 7n$ and $\abs{\omega} = 3n$ it suffices to show $\abs{x'} = \abs{yx} \geq n$.
			By $\mu y' = x' \mu$ we have that $\mu$ is a factor of $x'^+$. 
			We conclude $x'\not\in \Delta$ which implies $\abs{x'} > n$.
			In summary, $v = \omega v' \omega'$ and $\abs v \geq t$ holds. If $\abs v \leq t_\Omega$, then we can directly apply the $T_\Omega$-rule with left side $v$. Else, $v$ must be reducible by \prref{lem:tOmegaexist} and we can apply induction. 
		\end{proof}
		
		Combining the previous lemmas we show that $T$ is locally confluent.
		\begin{lemma}\label{lem:parikhisconfl}
			$T$ is locally confluent.
		\end{lemma}
		\begin{proof}		
			Let $\ell \to r, \ell' \to r' \in T$ be two rules. We have to show that every overlap of the left sides of those rules resolves.
			The system $T_\Delta$ is locally confluent by \prref{lem:deltacrsystem}. 
			Hence, we may assume that $\ell \to r\in T_\Omega$. Let $\omega u \omega' = \ell$ and consequently $r = \omega \gamma(u) \omega'$. 
			Consider first the case that $\delta^{t+n} = \ell' \to r' \in T_\Delta$.
			If $\ell'$ is a factor of $\ell$, that is, if $\ell = x\ell' y$, then $\ell \RAS{}{T} x r' y \RAS{*}{T} r$ by \prref{lem:parikhnf}. 
			By definition of $\Omega$, the left side $\ell$ which contains an element of $\Omega$ cannot be a factor of $\delta^{t+n}$. 
			Hence, the system resolves in the case of factor critical pairs.				
			Consider thus the case of an overlap critical pair $x \ell = \ell' y$ (the case $x \ell' = \ell y$ is symmetric). 
			Since $\omega$ is no factor of $\delta^+$ and $t \geq 3n$ by definition, we have the following situation:
			\begin{center}\begin{tikzpicture}[implies/.style={double,double equal sign distance,-implies}]
				\def\x{3.5}
				\def\y{9}
				\node [rectangle, thin, draw=black, inner sep = 1pt,minimum width = 1cm] at (\x-1,1) {$\delta^{n}\vphantom{\mu\delta^t}$};
				\node [rectangle, thin, draw=black, inner sep = 1pt,minimum width = 2cm] at (\x+0.5,1) {$\delta^{t}\vphantom{\mu\delta^t}$};
				\node [rectangle, thin, draw=black, minimum width = 1cm, inner sep = 1pt] at (\x+1.5,1.48) 
				{$\omega\vphantom{\mu\delta^t}$};
				
				\node [rectangle, thin, draw=black, minimum width = 1.5cm, inner sep = 1pt] at (\x+2.75,1.48) 
				{$u\omega'\vphantom{\mu\delta^t}$};
				\end{tikzpicture}
			\end{center}
			Let $\delta^t = z_1z_2$ and $\omega = z_2z_3$ be the overlap, then
			\begin{align*}
			x\ell \RAS{}{T} xr = x \omega \gamma(u) \omega' = \delta^{t+n} z_3\gamma(u) \omega' \RAS{}{T} \delta^t z_3 \gamma(u) \omega' = z_1z_2z_3 \gamma(u) \omega'\\
			\ell'y \RAS{}{T} r'y = \delta^ty = z_1 \omega u \omega' \RAS{}{T} z_1 \omega \gamma(u) \omega' = z_1z_2z_3 \gamma(u) \omega'
			\end{align*}
			
			Consider the case that $\ell' \to r' \in T_\Omega$ and let $\ell' = \mu v \mu'$. Again, if $\ell' = x\ell y$, then $\ell' \RAS{}{T} x r y \RAS{*}{T} r'$ by \prref{lem:parikhnf}. Hence, by symmetry, it suffices to consider the case $x \ell = \ell' y$. If $\ell$ and $\ell'$ overlap at most $3n$ positions, 
			\begin{center}\begin{tikzpicture}[implies/.style={double,double equal sign distance,-implies}]
				\def\x{3.5}
				\def\y{9}
				\node [rectangle, thin, draw=black, inner sep = 1pt,minimum width = 1.5cm] at (\x-0.25,1) {$\mu u'$};
				\node [rectangle, thin, draw=black, inner sep = 1pt,minimum width = 1cm] at (\x+1,1) {$\mu'$};
				\node [rectangle, thin, draw=black, minimum width = 1cm, inner sep = 1pt] at (\x+1.5,1.47) 
				{$\omega\vphantom{\mu\delta^t}$};  
				\node [rectangle, thin, draw=black, minimum width = 1.5cm, inner sep = 1pt] at (\x+2.75,1.47) 
				{$u\omega'\vphantom{\mu\delta^t}$};
				\end{tikzpicture}
			\end{center}
			then the rules can be applied independently; 
			let again be $\mu' = z_1z_2$ and $\omega = z_2z_3$ be the overlap, then
			\begin{align*}
			x\ell \RAS{}{T} xr = x \omega \gamma(u) \omega' = \mu u' \mu' z_3\gamma(u) \omega' \RAS{}{T} \mu \gamma(u') \mu' z_3 \gamma(u) \omega' = \mu \gamma(u')z_1z_2z_3 \gamma(u) \omega'\\
			\ell'y \RAS{}{T} r'y = \mu \gamma(u') \mu' y = \mu \gamma(u') z_1 \omega u \omega' \RAS{}{T} \mu \gamma(u') z_1 \omega \gamma(u) \omega' = \mu \gamma(u')z_1z_2z_3 \gamma(u) \omega'
			\end{align*}
			and the system resolves in this case. 
			
			Hence, we assume that $\ell$ and $\ell'$ overlap more than $3n$ positions. In this case $\mu'$ is a factor of $\ell$ and $\omega$ is a factor of $\ell'$. 
			This implies that $\mu$ and $\omega'$ are maximal $\Omega$-factors of $x\ell = \ell' y = \mu u'' \omega'$. We conclude $x\ell \RAS{}{T} xr \RAS{*}{T} \mu \gamma(u'') \omega'$ and $\ell' y \RAS{}{T} r' y \RAS{*}{T} \mu \gamma(u'')\omega'$ 
			by \prref{lem:parikhnf}. 
		\end{proof}
		
		By construction, the system $T$ is $\varphi$-invariant and thus the system $$T' = \set{c\ell \to cr \in A^*\times A^*}{\ell \to r \in T}$$ is $\varphi$-invariant. By \prref{lem:tOmegaexist} the system $T$ is of finite index over $K^*$. We can apply \prref{lem:bastel} and obtain a $\varphi$-invariant Parikh-reducing Church-Rosser system $S$ of finite index over $A^*$. 
		This concludes the proof of the first part of \prref{thm:parikhcomgroup}.
		It remains to study the groups in $A^*\! /S$.
		As an intermediate step, we study the groups in $K^*\! /T$.
		\begin{lemma}\label{lem:groupinT}
			Let $H \subseteq K^*\! /T$ be a subsemigroup which is a group and identify $H$ with the corresponding elements in $\IRR_T(K^*)$.
			Then either there exists some $\delta \in \Delta$ such that $H \subseteq \os{\delta^{t},\ldots, \delta^{t+n-1}}$ is a cyclic group whose order is divisible by $n$ or there is an injective homomorphism $\eta : H \to \prod_{a\in A} \mathbb Z/\ord(\varphi(a))\mathbb Z$.
		\end{lemma}
		\begin{proof}
		Without loss of generality, we may assume that $H$ is non-trivial.
			Let $e^2=e \in H$ be the identity element of $H$. 
			Note that by definition of the rules $T$ and the set $\Omega$, the irreducible word of every word $w\in K^*\Omega K^*$ also contains an $\Omega$-factor. 
			Thus, by $ex = x$ and $x^{\abs H} = e$ for all $x\in H$ either all elements in $H\subseteq K^*\! /T$ contain some factor in $\Omega$ or none of the elements contains an $\Omega$-factor. All words $x\in H$ must have length at least $t-n > 2n$ by definition of the rules $T$.
			
			Let us first consider the case that none of the elements contain an $\Omega$-factor. 
			We show that there exists some $\delta \in \Delta$ such that for all $x\in H$ there exists $i\in \mathbb N$ such that $x=\delta^i$.  
			Let $u\delta^{t+n} v \RA{T_\Delta} u \delta^t v$ be an application of a rule in $T_\Delta$ and let $w\in J$ be a minimal factor of $u\delta^{t+n}v$ which is not in $F$. 
			By \prref{lem:minimalefactorgegenbsp} $\abs w \leq 2n$ and since $t > 2n$, 
			the factor $w$ is also a factor of $u \delta^t v$. 
			Thus, the number of factors in $J$ does not decrease by an application of a rule in $T_\Delta$.
			Consider any $x \in H$. Since the number of factors in $J$ does not decrease by some application of a rule in $T_\Delta$, $x^{\abs{H}+1} = x$ and no rule in $T_\Omega$ is applicable, we deduce that the number of factors in $J$ of $x^{\abs{H}+1}$ and $x$ is the same. In particular, this number is zero and we obtain $x\in F$ for all $x\in H$. 
			Next, we show that $x = \delta^i$ for some $\delta \in \Delta$. 
			Since $x\in F$ and $\Delta$ is closed under conjugation, there exists a primitive word $\delta \in \Delta$ and $i\in \mathbb N$ such that $x = \delta^i \delta'$ for some prefix $\delta'$ of $\delta$. In particular, $\abs{\delta}$ is a period of $x$. 
			Note that $i \geq 2$ since $\abs{x} > 2n$.
			Consider the word $x^2$.
			By the above, we obtain $x^2 \in F$, that is, again there exists a primitive word $\hat \delta \in \Delta$, a prefix $\hat \delta'$ of $\hat \delta$ and a number $j\geq 2$ such that $x^2 = \hat \delta^j \hat \delta'$. 
			Therefore, $\abs{\hat \delta}$ is a period of $x^2$ and, hence, also of~$x$. 
			Since $\abs x > 2n$, we may use \prref{thm:finewilf} and conclude that $\gcd(\abs{\delta}, |\hat\delta|)$ is a period of $x$. 
			Since $\delta$ is primitive, this implies $\gcd(\abs{\delta}, |\hat\delta|) = \abs{\delta}$. Since $\hat{\delta}$ is a prefix of $x$, this yields that $\hat{\delta}$ is a power of $\delta$ which implies $\delta = \hat \delta$ by primitivity of $\hat{\delta}$.
			In particular, $\abs{\delta}$ is a period of $x^2$ and $\delta'\delta$ is a prefix of $\delta^2$. 
			Since $\delta$ is primitive this implies that $\delta'$ is not a proper prefix of $\delta$ by \prref{lem:charprimitive} and 
			we conclude that for every $x\in H$ there exists $\delta \in \Delta$ and $i\in \mathbb N$ such that $x = \delta^i$. 
			Thus, consider $\delta_1^i, \delta_2^j \in H$ with $\delta_1\neq \delta_2$ primitive words in $\Delta$. 
			Again, $\abs{\delta_1}$ is a period of $\delta_1^i$ and there must exist a period $p\leq n$ of $\delta_1^i\delta_2^j \in F$. 
			By \prref{thm:finewilf} $\gcd(\abs{\delta_1},p)$ is a period of $\delta_1^2$.
			By primitivity of $\delta_1$, this yields that $\abs{\delta_1}$ is a divisor of $p$. In particular, since $p$ is a period of $\delta_1^i \delta_2^j$, this yields $\delta_1^i\delta_2^j = \delta_1^i\delta_1^k \delta_1'$ for some $k\geq 2$ and $\delta_1'$ a prefix of $\delta_1$. 
			Using \prref{thm:finewilf} again, we see that $\gcd(\abs{\delta_1},\abs{\delta_2})$ is a period of $\delta_2^j$, that is, $\abs{\delta_2}$ is a divisor of $\abs{\delta_1}$ by primitivity of $\delta_2$. By symmetry, this yields $\abs{\delta_1} = \abs{\delta_2}$ and thus $\delta_1 = \delta_2$.
			
			Fix some primitive word $\delta \in \Delta$ such that $H \subseteq \delta^+$. Since $ex = x$ for all $x\in H$ and the right side of rules in $T_\Delta$ have length at least $t$ and since $\delta^{t+n}$ is reducible, we conclude $H \subseteq \os{\delta^t,\ldots, \delta^{t+n-1}}$ and thus $H$ is a subgroup of the cyclic group $\os{\delta^t,\ldots, \delta^{t+n-1}}$ of order $n$ which finishes this case.

			The second case is that all words in $H$ contain an $\Omega$-factor. 
			Consider the maximal $\Omega$-factors of $e$ and factorize 
			$e = e_1 \omega e_2 \omega' e_3$ with $\omega, \omega' \in \Omega$ maximal for $e$ 
			such that $e_1\omega$ and $\omega' e_3$ contains no other maximal $\Omega$-factors of $e$. 
			Since $e^2 = e$, we conclude that $e_2$ is some normal form. 
			By $ex = x = xe$ for all $x\in H$ and \prref{lem:parikhnf}, 
			there must exist a factorization $x = e_1 \omega \hat x \omega' e_3$ such that $\hat x = \gamma(\hat x)$ is a normal form. 
			In particular, $\widehat{xy} = \gamma(\hat x \omega' e_3 e_1 \omega \hat y)$ by \prref{lem:parikhnf}.
			Consider the homomorphism $\psi : A^* \to \prod_{a\in A} \mathbb Z/\ord(\varphi(a))\mathbb Z$ 
			which counts the number of $a\in A$ modulo $\ord(a)$ and the function $\eta : H \to \prod_{a\in A} \mathbb Z/\ord(a)\mathbb Z$ 
			given by $\eta(x) = \psi(\hat x) \cdot \psi(\omega'e_3e_1\omega)$. 
			Note that $\psi(\widehat{xy}) = \psi(\hat x)\psi(\hat y) \psi(\omega'e_3e_1\omega)$ implies that
			$\eta$ is a homomorphism. 
			It holds $\eta(x) = \eta(y)$ if and only if $\psi(\hat x) = \psi(\hat y)$. 
			By definition of the normal forms $\gamma(\cdot)$, it holds $\psi(\hat x) = \psi(\hat y)$ if and only if $\hat x = \hat y$ 
			and therefore $\eta$ is injective.
		\end{proof}
		By \prref{lem:bastel}, we obtain that the subgroups in $A^*\! /S$ are isomorphic to subgroups of $B^*\! /R$ and $K^*\! /T$. 
		By induction, all groups in $B^*\! /R$ are  isomorphic to some subgroup of $\prod_{a\in A} \mathbb Z/\ord(\varphi(a))\mathbb Z$. 
		All groups in $K^*\! /T$ are either cyclic of order divisible by $n$ or isomorphic to some subgroup of $\prod_{a\in A} \mathbb Z/\ord(\varphi(a))\mathbb Z$ by \prref{lem:groupinT}. 
		However, since $n$ is defined as the least common multiple of $\ord(\varphi(a))$, 
		the cyclic group of order $n$ is a subgroup of $\prod_{a\in A} \mathbb Z/\ord(\varphi(a))\mathbb Z$. 
		This proves the statement.
	\end{proof}
	
	\subsection{Group languages over an alphabet of size two}\label{subsec:group2generator}
	The same technique as in \prref{subsec:commgrpcr} can be used to obtain Parikh-reducing Church-Rosser systems which factorize through homomorphisms $\varphi : \oneset{a,b}^* \to G$ for an arbitrary group $G$. We will only sketch the proof, as it is essentially the proof of \prref{thm:parikhcomgroup}. 
	\begin{theorem}
		Let $A = \oneset{a,b}$ be an alphabet of size two and let $\varphi : A^* \to G$ be a homomorphism into a finite group $G$. Then there exists a Parikh-reducing Church-Rosser system $S$ of finite index which factorizes through $\varphi$. All groups in $A^*/S$ are subgroups of $G$ or of $\mathbb Z/n\mathbb Z$ where $n$ is the exponent of $G$.
	\end{theorem}
	\begin{proof}[Sketch of proof]
		Let $n$ be the exponent of $G$ and let $R = \oneset{a^n \to 1} \subseteq \oneset{a}^* \times \oneset{a}^*$ be the set of rules over the alphabet $\oneset{a}$. Set $K = \IRR_R(a^*)b = \set{a^ib}{0 \leq i < n}$. 
		In the remainder of the sketch, we have to construct a system over $K^*$. 
		As the set of short words we choose $\Delta = K^{\leq n^2} \setminus \oneset{1}$. The corresponding set of rules is $T_\Delta = \set{\delta^{t+n} \to \delta^t}{\delta \in \Delta}$ for $t = n^2(3n+7)$. Note that since $t > 2n^2$ the system $T_\Delta$ is confluent by \prref{lem:deltacrsystem}. 
		
		Let $F =  \bigcup_{\delta \in \Delta, i \in \mathbb N} \Factors(\delta^i)$ and set $\Omega = K^{3n^2} \setminus (bK^*\cup F)$. Choose a preorder $\preceq$ on $\Omega$ such that 
		\begin{itemize}
			\item $\omega, \eta \in\Omega$ with 
			$\omega \in K^*(K\setminus\os{b}) b^i, \eta \in K^*(K\setminus\os{b}) b^j$ and $i>j$ 
			implies $\omega \preceq \eta$.
			\item $\preceq$ is a total order on $\Omega \setminus Kb^{3n^2-1}$.
			\item $\omega, \eta \in \Omega\cap Kb^{3n^2-1}$ implies $\omega \preceq \eta$.
		\end{itemize}
		
		In order to complete the construction, it remains to choose the normal forms $v_g$. Note that every representation of $g\in G$ needs less than $n$ a's by the pigeonhole principle. Thus, for every $g\in G$ there exists a word $v_g = b^{3n^2} v_{1} b^{3n^2} \cdots b^{3n^2} v_{n-1} b^{3n^2} \in K^*$ with $\varphi(v_g) = g$ and $v_i \in \set{ab^k, b^k}{1\leq k \leq n}$. For every $g \in G$ we choose such a word $v_g$ such that the number of $a$'s is minimal. Note that by construction $\abs{\abs{v_g}-\abs{v_h}} < n^2$ as a word over $K$. This is the reason for the choice of $\Delta$. Furthermore, $t-7n^2 < \abs{v_g} < t-6n^2$, which explains the choice of the parameter $t$. 
		The choice of $v_g$ also yields that there are no $\Omega$-factors in $v_g$ apart from $ab^{3n^2}$, which is $\Omega$-minimal.
		
		Adapting the proof of \prref{lem:t0parikhexist}, we prove the existence of a number $t_0$ such that every word $v \in K^*$ of length at least $t_0$ has a factor $\delta^{t+n}$ for a $\delta \in \Delta$ or a factor $\omega \in \Omega$. \prref{lem:tOmegaexist} yields the existence of a number $t_\Omega$ such that every $v \in \IRR_{T_\Delta}(K^*)$ contains a factor $\omega \, u \, \omega'$ 
		with $\omega, \omega'$ being $\Omega$-maximal for this factor and $t< \abs{\omega u \omega'} \leq t_\Omega$. 
		Again, let \[T_\Omega = \set{\omega u \omega' \to \omega v_{\varphi(u)} \omega'}{t\leq \abs{\omega u \omega'} \leq t_\Omega \text{ and } \omega, \omega' \text{ are }\Omega\text{-maximal for } \omega u \omega'}\] and $T = T_\Delta \cup T_\Omega$. 
		We want to apply \prref{lem:bastel} to obtain a system $S \subseteq \oneset{a,b}^* \times \oneset{a,b}^*$. 
		Confluence of $T$ follows along the lines of \prref{lem:omegamax}, \prref{lem:parikhnf} and \prref{lem:parikhisconfl}, whereas the statement about the groups in $A^*/S$ is analogously to \prref{lem:groupinT}.
	\end{proof}
	
	\section{Beyond Groups}\label{sec:beyondgroups}
	In this section we apply local divisors in order to lift the construction of Church-Rosser systems for groups to the general case of monoids. Instead of directly constructing a system over $K = \IRR_R(B^*)c$, we obtain a system inductively by going over to the local divisor. This decreases the size of the monoid, but increases the size of alphabet. The first part of this theorem has been published in \cite{DiekertKRW15jacm}, whereas the second part is based on the use of Rees extensions, see \cite{DiekertKW12tcs, DiekertWalter16}.
	\begin{theorem}\label{thm:group2variety}
		Let $\varietyH$ be a group variety such that for every homomorphism $\varphi : A^* \to G$ for $G\in \varietyH$ there exists a Parikh-reducing Church-Rosser system $S$ of finite index which factorizes through $\varphi$.
		Let $\varphi : A^* \to M$ be a homomorphism with $M \in \varietyHline$.
		\begin{enumerate}
			\item\label{enum:group2variety:a} There exists a $\varphi$-invariant Parikh-reducing Church-Rosser system $S$ of finite index.
			\item\label{enum:group2variety:b}  If every homomorphism $\varphi : A^* \to G$ in a group $G \in \varietyH$ has a Church-Rosser representation in $\varietyHline$, then $A^*\! /S \in \varietyHline$.
		\end{enumerate}
	\end{theorem}
	\begin{proof}
		\ref{enum:group2variety:a}. We use induction on $(\abs{M}, \abs{A})$, ordered lexicographically. Since $\varietyHline$ is closed 
		under taking submonoids, we can restrict ourselves on surjective homomorphisms $\varphi$. 
		If $M$ is a group, then $M\in \varietyH$ and there exists such a system $S$ by the preconditions. 
		Thus, we can assume that there is 
		a letter $c\in A$ such that $\varphi(c)$ is not a unit. Let $B = A\setminus \oneset{c}$. 
		By induction the restriction 
		\[
		\varphi|_{B^*} : B^*\to M
		\]
		admits a Parikh-reducing Church-Rosser system $R\subseteq B^*\times B^*$. 
		Consider the set 
		\[
		K = \IRR_R(B^*)c.
		\]
		This is a prefix code and will be considered as a new alphabet. 
		Let $\psi : K^* \to M_{\varphi(c)}$ be the homomorphism to the local divisor at $\varphi(c)$ induced via 
		$\psi(uc) = \varphi(cuc)$. We have $\abs{M_{\varphi(c)}}  < \abs{M}$ and $M_{\varphi(c)} \in \varietyHline$ and thus, by induction, there exists a 
		Parikh-reducing Church-Rosser system $T'\subseteq K^* \times K^*$ of finite index, 
		such that $T'$ factorizes through $\psi$.
		In particular, we have $\psi(\ell) = \psi(r)$ for a rule $(\ell, r) \in T'$. 
		We show that $\varphi(c\ell) = \varphi(cr)$. For this let $\ell = u_1c \ldots u_nc$ and $r = v_1c \ldots v_mc$. It holds
		\begin{align*}
		\varphi(c\ell) &= \varphi(cu_1c) \circ \ldots \circ \varphi(cu_nc)\\
		&= \psi(u_1c)\circ \ldots \circ \psi(u_nc) \\
		&= \psi(\ell) = \psi(r) \\
		&= \psi(v_1c)\circ \ldots \circ \psi(v_mc) \\
		&= \varphi(cv_1c) \circ \ldots \circ \varphi(cv_mc) = \varphi(cr).
		\end{align*}
		Hence, the rule $c\ell \to cr$ is $\varphi$-invariant. We set
		\[
		T = \set{c\ell \to cr}{\ell \to r \in T'}.
		\]
		The system $S = R\cup T$ has the required properties by \prref{lem:bastel}.
		
		\ref{enum:group2variety:b}. 
		The statement is clear if $M$ is a group. Consequently, the construction above is applied.
		By induction we may assume that $B^*\! /R, K^*\! /T \in \varietyHline$ and \prref{lem:bastel} implies that $A^*\! /S \in \varietyHline$.
	\end{proof}	
	A direct combination of \prref{thm:parikhcomgroup} and \prref{thm:group2variety} yields the following corollary.
	\begin{corollary}
		Let $M \in \overline{\Gcom}$ be a monoid and $\varphi : A^* \to M$ be a homomorphism, then there exists a Parikh-reducing Church-Rosser system $S\subseteq A^*\times A^*$ such that $S$ factorizes through $\varphi$ and $A^*\! /S \in \overline{\Gcom}$. 
		In particular, every language $L\subseteq A^*$ recognized by $\varphi$ is given as a finite union $L = \bigcup_{u\in L} [u]_S$.
	\end{corollary}

	In particular, \prref{thm:group2variety} shows that one can control the groups in the Church-Rosser representation. However, in general one may not preserve other properties, for instance, commutativity.   
	\begin{proposition}
		Let $\varphi : A^* \to \mathbb Z/2\mathbb Z$ be the homomorphism mapping each letter to the generator of $\mathbb Z/2\mathbb Z$. If $\abs{A}>1$, there is no abelian Church-Rosser representation of~$\varphi$.
	\end{proposition}
	\begin{proof}
		Assume that there exists a Church-Rosser system $S$ of finite index such that $A^*/S$ is abelian and there exists a homomorphism $\psi : A^*/S \to \mathbb Z/2\mathbb Z$ with $\varphi = \psi \circ \pi_S$. Let $a,b \in A$ be letters such that $a \neq b$. Since $S$ factorizes through $\varphi$, we have $\abs{r} \equiv \abs{\ell} \bmod 2$ for every rule $(\ell,r) \in S$ and it holds $a\neq b$ in $A^*/S$. Since $A^*/S$ is abelian, we obtain $ab = ba$ in $A^*/S$. In particular, $ab \to_S 1 \leftarrow_S ba$ and $A^*/S$ must be a group. Let $2n$ be the order of $a$ and $b$. Then $a^n = a^nb^nb^{n} = b^n$ holds in $A^*/S$ and thus there must be a irreducible word $w$ with $a^n \RAS{*}{S} w \LAS{*}{S} b^n$. By the argumentation above, there exists a number $k<n$ such that $w \in \oneset{a^k,b^k}$. Thus, either $a^{n-k} = 1$ or $b^{n-k} = 1$ which is a contradiction to the definition of $n$.
	\end{proof}
	\section{Complexity of Church-Rosser systems}\label{sec:crcomplexity}
	In this section we analyze the size of a Church-Rosser representation as constructed by \prref{thm:group2variety} and \prref{thm:parikhcomgroup}. We will restrict our analysis on the construction of the Parikh-reducing Church-Rosser representation. Similiar calculations can be made for the analysis of the size of the Church-Rosser system. 
	
	Before we prove upper bounds for the size of the constructed Church-Rosser systems, we reconsider the construction. 
	Our constructions used \prref{lem:bastel} as the basic building block of the construction. 
	Let $\varphi : A^* \to M$ be a homomorphism.
	For $B = A\setminus\os{c}$ and a system $R\subseteq B^*\times B^*$ one needs a system $T\subseteq K^* \times K^*$ for the alphabet $K=\IRR_R(B^*)c$.
	Now, unlike in the general case, we are able to reduce the alphabet itself by exploiting the structure of the alphabet. 
	Let $b_1\cdots b_k c \in K$ with $b_i \in B$ and $k>\abs{M}$. By the pigeonhole principle there exist $i<j$ such that $\varphi(b_1\cdots b_i) = \varphi(b_1\cdots b_j)$ and ${i+(k-j)\leq n}$. 
	Thus, we may introduce the subword-reducing\footnote{subword-reducing seen as a rule over $A^*$, not over $K^*$.} rule $b_1\cdots b_k c \to b_1\cdots b_i b_{j+1}\cdots b_k c$. 
	If $b_1\cdots b_i b_{j+1}\cdots b_k$ is reducible in $R$, reduce it further in $R$. Repeating this process yields a new alphabet for $K$ which is a subset of $B^{\leq n}c$ and therefore, if $\abs{B}>1$, has at most $(\abs{B}^{n+1}-1)/(\abs{B}-1)$ elements. One can check, that the proofs of 
	\prref{thm:group2variety} and \prref{thm:parikhcomgroup} also work adding this reduction technique of the alphabet $K$. 
	We refrained from directly adding it to the theorems, as they are already quite technical. 
	
	\begin{proposition}\label{prop:complexitygroup}
		Let $\varphi : A^* \to G$ be a homomorphism in $G\in \Gcom$, $n=\abs{G}$ and $m = \abs{A}>1$,  then there exists a Parikh-reducing Church-Rosser system $S$ such that $S$ factorizes through $\varphi$ and
		$$\abs{A^*\! /S} \in 2^{2^{m^{\mathcal O(n^2)}}}.$$
	\end{proposition}
	\begin{proof}
		Let $S$ be the Parikh-reducing Church-Rosser system constructed using \prref{thm:parikhcomgroup} and the reduction technique described above.
		\prref{lem:bastel} shows that for $m>1$ it holds $$\abs{A^*\! /S} = \abs{B^*\! /R} + \abs{B^*\! /R}^2\cdot  \abs{K^*\! /T} \leq 2\abs{B^*\! /R}^2\cdot  \abs{K^*\! /T}$$ where $B=A\setminus \os{c}$. 
		In the case of \prref{thm:parikhcomgroup}, $R$ is constructed inductively whereas $T$ is constructed directly. 
		By \prref{lem:tOmegaexist}, every irreducible word in $\IRR_T(K^*)$ has length at most $t_\Omega$ and therefore $\abs{K^*\! /T} \leq \abs{K}^{t_\Omega}$. The construction of $t_\Omega$ in the proof of \prref{lem:tOmegaexist} shows that $t_\Omega \leq 2^{\abs{\Omega}}(t_0+t)$ whereas $t_0 + t \in \mathcal O(n^2m)$. Since $\Omega \subseteq K^{3n}$ we obtain $$\abs{K^*\! /T} \leq \abs{K}^{\mathcal O(n^2m) \cdot 2^{\abs{K}^{3n}}}.$$
		Using the alphabet reduction technique, we can assume $\abs{K} \leq m^{n+1}$. 
		Note that $\abs{K}^{3n} \leq  (m^{n+1})^{3n} = m^{(n+1)3n}$ does not yield another exponential jump. 
		A straightforward calculation yields the existence of a constant $c\in \mathbb N$ such that $$2\abs{K^*\! /T} \leq 2^{2^{m^{cn^2}}}.$$ 
		Now let $\ms(\varphi)$ denote the smallest size of a Parikh-reducing Church-Rosser representation of $\varphi$ and set 
		\begin{align*}
		\ms(n,m) = \max\set{\ms(\varphi)}{\varphi : A^* \to G, \abs{A}\leq m, G\in \Gcom, \abs{G}\leq n}
		\end{align*}
		to be the complexity over all possible homomorphisms with $\abs{A}\leq m$ and $\abs{G}\leq n$. 
		We have seen that the recursion
		\begin{align*}
		\ms(n,m) \leq \ms(n,m-1)^2 \cdot 2^{2^{m^{cn^2}}}
		\end{align*} 
		holds and show $\ms(n,m) \leq 2^{2^{m^{cn^2+2}}}$ inductively using this recursion. 
		Note that $\ms(n,1) = n$ and thus the inequality is true in the base case $m=2$. 
		Also $\ms(1,m) = 1$ and therefore we assume $n>1$.
		For $m>2$ and $n>1$ it holds
		\begin{align*}
		\ms(n,m) &\leq \ms(n,m-1)^2 \cdot 2^{2^{m^{cn^2}}}\\
		&\leq 2^{2^{(m-1)^{cn^2+2}+1}} \cdot 2^{2^{m^{cn^2}}}\\
		&= 2^{2^{(m-1)^{cn^2+2}+1} + 2^{m^{cn^2}}} \\
		&\leq 2^{2^{(m-1)^{cn^2+2}+1 + m^{cn^2}}}\\
		&\leq 2^{2^{m^{cn^2+2}}}.
		\end{align*}
		The last inequality holds since 
		\begin{align*}
		(m-1)^{cn^2+2}+1 + m^{cn^2} &\leq (m-1)m^{cn^2+1} + 1 + m^{cn^2}\\
		&= m^{cn^2+2} + m^{cn^2}\underbrace{(1-m)}_{<0} + 1 \\
		&\leq m^{cn^2+2}. \qedhere
		\end{align*}
	\end{proof}
	
	The triple exponential upper bound given by \prref{prop:complexitygroup} seems huge, however there is already a single exponential lower bound which is fairly easy to see. The lower bound comes from the fact that Church-Rosser systems cannot directly represent group identities which preserve length, such as commutation.
	\begin{proposition}\label{prop:lowerbound}
		For every $n\in \mathbb N$ there exists a homomorphism $\varphi : A^* \to G$ into an abelian group $G$ of size $n$ such that for every length-reducing Church-Rosser system $S$ which factorizes through $\varphi$ all words of length smaller than $n$ are irreducible, that is, $A^{<n} \subseteq \IRR_S(A^*)$. In particular, if $\abs A > 1$:
		$$\abs{A^*\! /S} \geq (\abs{A}^n-1)/(\abs{A}-1).$$
	\end{proposition}
	\begin{proof}
		Consider the cyclic group $G$ of order $n$ and the homomorphism $\varphi : A^* \to G$ which maps all letters $a\in A$ to the same generator $g$ of $G$. Let $S \subseteq A^*\times A^*$ be a length-reducing Church-Rosser system which factorizes through $\varphi$. We show that every word of length less than $n$ is irreducible in $S$. Let $w \in A^*$ be a word with $\abs{w} < n$. 
		Assume that $w \RAS{}{S} v$ for some word $v$. Since $S$ is length-reducing, $\abs{v} < \abs{w}$. 
		However, $\varphi(w) = \varphi(v)$ implies $g^{\abs{w}-\abs{v}} = 1$. Since the order of $g$ is $n$, this is a contradiction to $0 < \abs{w}-\abs{v} < n$ and $w$ must be irreducible.
	\end{proof}
	Note that this proof does not use the Church-Rosser property and thus one could expect a larger size of the Church-Rosser representation. 
	
	\begin{example}
		Niemann and Waldmann constructed an explicit Parikh-reducing system $S$ for the case $\varphi : A^* \to \mathbb Z/2\mathbb Z$ with $\varphi(a) = 1$ for all $a\in A$ \cite{NiemannW02,NiemannPhD02}. Their system is given by $S = \set{xyz \to \max(x,z)}{x,y,z\in A, y = \min(x,y,z)}$ for some arbitrary order on~$A$. 
		The irreducible elements in $A^*\! /S$ are exactly the sequences which are first strictly increasing and then strictly decreasing, that is $$\IRR_S(A^*) = \set{a_1 \cdots a_i \cdots a_n}{a_1 < \cdots < a_i \geq a_{i+1} > \cdots > a_n}.$$ This yields $\abs{A^*\! /S} = \abs{\IRR_S(A^*)} = 1+ \sum_{i=1}^{\abs{A}} 2^{2i-1} = (2^{2\abs{A}+1}+1)/3$ which is significantly larger than the lower bound $\abs{A} + 1$ given in \prref{prop:lowerbound}.
	\end{example}
	In the monoid case, the minimal size of a Church-Rosser representation is bounded by a quadruple exponential function. This increase in complexity, compared to the group case, comes from the fact that, unlike in the group case, the system $T \subseteq K^* \times K^*$ is constructed by induction. However, this is also the reason that the alphabet reduction technique is even more powerful in this case. 
	Consider the function $f : \mathbb N^2 \to \mathbb N$ given by $f(1,m) = 1$, $f(n,1) = n$ and $f(n,m) = 2f(n,m-1)^2 \cdot f(n-1, f(n,m-1))$ for $n,m>1$. This function gives an upper bound for the maximal size of a Church-Rosser representation of a monoid of size $n$ and an alphabet of size $m$ without any optimization. Consider further the hyperoperation function $A_1(n) = 2n$, $A_k(1) = 2$ and $A_k(n) = A_{k-1}(A_k(n-1))$.\footnote{The notation $A$ comes from Ackermann, since the function $A$ is a modified Ackermann function.} 
	For fixed $k$, the function $A_k$ is primitive recursive, however the two-variable function $A$ grows faster than any primitive recursive function, see e.g. \cite{DavisWeyuker83}. An induction shows that $f(n,m) \geq A_{n-1}(m)$ for $n>1, m\geq 1$. 
	Hence, without the alphabet reduction the recursive formula would yield a non-primitive recursive function.
	\begin{proposition}
		Let $\varphi : A^* \to M$ be a homomorphism in $M\in \overline{\Gcom}$, $n=\abs{M}$ and $m = \abs{A}$. Then there exists a Parikh-reducing Church-Rosser system $S$ such that $S$ factorizes through $\varphi$ and
		$$\abs{A^*\! /S} \in 2^{2^{m^{\mathcal O\left((n+1)!\right)} + n}}.$$
	\end{proposition}
	
	\begin{proof}
		If $M \in \Gcom$, we know that there exists such a system $S$ with $\abs{A^*\! /S} \in 2^{2^{m^{\mathcal O(n^2)}}}$ by \prref{prop:complexitygroup}. 
		If $m=1$, then there exists a system $S$ such that $\abs{A^*\! /S} \leq n$. 
		In the other case we will use the local divisor construction of \prref{thm:group2variety}. 
		Note that by the alphabet reduction technique we may assume that $\abs{K} < m^{n+1}$. 
		
		Let $\ms(\varphi)$ denote the smallest size of a Parikh-reducing Church-Rosser representation of $\varphi$ and set 
		\begin{align*}
		\ms(n,m) = \max\set{\ms(\varphi)}{\varphi : A^* \to M, \abs{A}\leq m, M\in \overline{\Gcom}, \abs{M}\leq n}
		\end{align*}
		to be the complexity over all possible homomorphisms with $\abs{A}\leq m$ and $\abs{M}\leq n$. 
		
		The base cases are $m=1$ or $M$ is a group. For $m=1$ there exists a system of size~$n$. 
		In all other cases we have the following recursion formula for $\ms(n,m)$:
		\begin{align*}
		\ms(n,m) \leq 2\ms(n,m-1)^2 \cdot \ms(n-1,m^{n+1}).
		\end{align*}
		Note that $n>1$ since $M$ is not a group.
		Choose $c \in \mathbb N$ such that $\ms(n,m) \leq 2^{2^{m^{c(n+1)!}+n}}$ for all base cases.  This is possible since the group case is in $2^{2^{m^{\mathcal O(n^2)}}}$. We show that $$\ms(n,m) \leq 2^{2^{m^{c(n+1)!}+n}}$$ in general. 
		Inductively, it holds 
		\begin{align*}
		\ms(n,m) &\leq 2\ms(n,m-1)^2 \cdot \ms(n-1,m^{n+1})\\
		&\leq 2\cdot 2^{2^{(m-1)^{c(n+1)!}+n+1}} \cdot 2^{2^{(m^{n+1})^{cn!}+n-1}}\\
		&= 2^{1 + 2^{(m-1)^{c(n+1)!}+n+1} + 2^{m^{c(n+1)!}+n-1}} \\
		&\leq 2^{2^{m^{c(n+1)!}+n}}.
		\end{align*}
		The last inequality holds because for $n,m>1$
		\begin{align*}
		(m-1)^{c(n+1)!} &\leq (m-1)^2\cdot m^{c(n+1)!-2} \\
		&= m^{c(n+1)!} - (2m-1)m^{c(n+1)!-2} \\
		&\leq m^{c(n+1)!}-3
		\end{align*}
		and thus $(m-1)^{c(n+1)!}+n+1 < m^{c(n+1)!}+n-1$.
	\end{proof}
	
	\section{Conclusion}
	In this paper we introduced the notion of Parikh-reducing Church-Rosser representations. We were able to construct such representations in the case of languages in $\overline{\mathbf{Ab}}$ and for group languages over a two-element alphabet. Furthermore, we studied algebraic properties of such representations and the complexity of the corresponding systems. Several questions remain open as future work. Most importantly, 
	does there exist a finite Parikh-reducing Church-Rosser representation for every homomorphism into a finite group? Note that this already implies the case for every finite monoid by \prref{thm:group2variety}.
	Another interesting open question is which algebraic properties can be preserved by Church-Rosser representations. For example, it seems unlikely that every homomorphism into a finite group has a Church-Rosser representation which is a group again, although it may happen in some special cases.
	Additionally, there is a huge gap between our lower and upper bounds for the complexity. 
	Therefore it is interesting whether there are constructions for Church-Rosser representations which yield a better upper bound and what a good lower bound for the size of a Church-Rosser representation is.
	
\newcommand{\Ju}{Ju}\newcommand{\Ph}{Ph}\newcommand{\Th}{Th}\newcommand{\Ch}{Ch}\newcommand{\Yu}{Yu}\newcommand{\Zh}{Zh}\newcommand{\St}{St}\newcommand{\curlybraces}[1]{\{#1\}}

\end{document}